\documentclass[conference]{IEEEtran}

\usepackage{cite}
\ifCLASSINFOpdf
   \usepackage[pdftex]{graphicx}
   \graphicspath{{figs/}}
  
   \DeclareGraphicsExtensions{.pdf,.jpeg,.png}
\else  
   \usepackage[dvips]{graphicx}
   \graphicspath{{./figs/}}  
   \DeclareGraphicsExtensions{.eps}
\fi

\usepackage[cmex10]{amsmath}
\usepackage{algorithmic}
\usepackage{array}

\usepackage{subfig}
\usepackage{float}
\usepackage{caption}
\usepackage{amssymb}
\usepackage{multirow}
\usepackage{booktabs}
\usepackage{mdwmath}
\usepackage{mdwtab}
\usepackage{bbm}
\usepackage{dsfont}
\usepackage{enumerate}
\usepackage{fixltx2e}
\usepackage{stfloats}

\usepackage{url}
\newtheorem{theorem}{Theorem}
\newtheorem{lemma}{Lemma}

\newtheorem{proposition}{Proposition}
\begin{document}

\title{Finding Community Structure with Performance Guarantees in Complex Networks}

\author{\IEEEauthorblockN{Thang N. Dinh and  My T. Thai}
\IEEEauthorblockA{Computer \& Information Science \& Engineering\\
University of Florida, Gainesville, FL, 32611,\\
Email: \{tdinh,mythai\}@cise.ufl.edu}
}

\maketitle

\begin{abstract}
Many networks including social networks, computer networks, and biological networks are found to divide naturally into communities of densely connected individuals. Finding community structure is one of fundamental problems in network science. Since Newman's suggestion of using \emph{modularity} as a measure to qualify the goodness of community structures, many efficient methods to maximize modularity have been proposed but without a guarantee of optimality. In this paper, we propose two polynomial-time algorithms to the modularity maximization problem with theoretical performance guarantees. The first algorithm comes with a \emph{priori guarantee} that the modularity of found community structure is within a constant factor of the optimal modularity when the network has the power-law degree distribution. Despite being mainly of theoretical interest, to our best knowledge, this is the first approximation algorithm for finding community structure in networks. In our second algorithm, we propose a \emph{sparse metric}, a substantially faster linear programming method for maximizing modularity and apply a rounding technique based on this sparse metric with a \emph{posteriori approximation guarantee}. Our experiments show that the rounding algorithm returns the optimal solutions in most cases and are very scalable, that is, it can run on a network of a few thousand nodes whereas the LP solution in the literature only ran on a network of at most 235 nodes.
\end{abstract}
\section{Introduction}
Many complex systems of interest such as the Internet, social, and biological relations, can be represented as networks consisting a set of \emph{nodes} which are connected by \emph{edges} between them. Research in a number of academic fields has uncovered unexpected structural properties of complex networks including small-world phenomenon \cite{Watts98}, power-law degree distribution \cite{Barabasi00}, and the existence of community structure \cite{Milo02} where nodes are naturally clustered into tightly connected modules, also known as communities, with only sparser connections between them.

The detection of community structures in networks is an important problem that has drawn an enormous amount of research effort \cite{Fortunato08}. A huge benefit of identifying community structure is that one can infer semantic attributes for different communities. For example in social networks, the attributes for a community can be common interest or location, and for metabolic networks the attribute could be a common function. Moreover, the relative independence among different communities allows the examining of each community individually, and an analysis of network at a higher-level of structure.

There are a wide variety of definitions for communities. In general, definitions can be  classified into two main categories: \emph{local definitions} and \emph{global definitions}. In local definitions, only the group of nodes  and its immediate neighborhood are considered, ignoring the rest of the network. For example, communities can be defined as maximal \emph{cliques}, \emph{quasi-cliques}, $k$-\emph{plexes}. The most famous definitions in this category are notions of \emph{strong community}, where each node has more neighbors inside than outside  the community, and \emph{weak community}, where the total number of inner edges must be at least half of the number of outgoing edges. 

In global definitions, communities can be only recognized by analyzing the network as a whole. This type of definitions is especially suitable when the next phase after the community detection is to optimize a global quantity, for example, minimizing the inter-group communication cost. The most widely-used quantity function in the global category is Newman's \emph{modularity} which is defined as the number of edges falling within communities minuses the expected number in an equivalent network with edges placed at random \cite{Girvan02}. A higher value of modularity, a better community structure. Thus, identifying a good community structure of a given network becomes finding a partition of networks so as to maximize the modularity of this partition, called modularity maximization problem.

Since the introduction of modularity, maximizing modularity has become primal approaches to detect community structure. Numerous computational methods have been proposed, based on agglomerative hierarchical clustering\cite{Day84}, simulated annealing\cite{Reichardt06}, genetic search \cite{Anca07}, extremal optimization \cite{Duch05}, spectral clustering \cite{Newman06}, multilevel partitioning \cite{Blondel08}, and many others. For a comprehensive view of community detection methods, we refer to an excellent survey of S. Fortunato and C. Castellano \cite{Fortunato08}. 

Unfortunately, Brandes et al. \cite{Brandes08} have shown that modularity maximization is an NP-hard problem, thereby denying the existence of polynomial-time algorithms to find optimal solutions. Thus, it is desirable to design polynomial-time approximation algorithms to find partitioning with a theoretical performance guarantee on the modularity values.

In contrary to the vast amount of work on maximizing modularity, the only known polynomial-time approach to find a good community structure with guarantees is due to G. Agarwal and D. Kempe \cite{Agarwal08} in which they rounded the fractional solution of a linear programming (LP). The value obtained by the LP is an upper bound on the maximum achievable modularity. Thus, their approach provide a posteriori guarantee on the error bound. In fact, the modularity values found by their approach are optimal for many network instances comparing with the optimal modularity values provided by expensive exact algorithms in \cite{Cafieri10}. The main drawback of the approach is the large LP formulation that consumes both time and memory resources. As shown in their paper, the approach can only be used on the networks of up to 235 nodes. Secondly, while the approach performs well on all considered networks, it does not promise any priori guarantees as provided by \emph{approximation algorithms}.

In this paper, we address the main drawback of the rounding LP approach by introducing an improved formulation, called \emph{sparse metric}. We show that our new technique substantially reduces the time and memory requirements  both theoretically and experimentally without any trade-off on the quality of the solution. The size of solved network instances raises from hundred to several thousand nodes while the running time on the medium-instances are sped up from 10 to 150 times.  

Our second contribution is an approximation algorithm that finds a community structure in networks with modularity values within a constant factor of the optimum when the considered networks have power-law degree distributions. To our best knowledge, it is the first approximation algorithm for finding community structure in networks. The algorithm is not only of theoretical interest, but also establish a connection between the power-law degree distribution properties and the presence of community structure in complex networks. Since community structure are often observed together with the power-law property, studying the community structure detection under power-law network models is of great important.  
  
\textbf{Organization.} We present definitions and notions in Section \ref{sec:prem}. We propose in Section \ref{sec:lp} the \emph{sparse metric} technique to efficiently maximize modularity via rounding a linear programming. An approximation algorithm for networks with the power-law degree distribution (so-called power-law networks) is introduced in Section \ref{sec:approx}. We show experimental results for the \emph{sparse metric} in Section \ref{sec:exp} to illustrate the time efficiency over the previous approach. Finally, in Section \ref{sec:dis} we summarize our results and  discuss on limitation of modularity as well as the corresponding resolution.

\section{Preliminaries}
\label{sec:prem}
A network can be represented as an undirected graph $G=(V,E)$\ consisting of $n=|V|$ nodes and $m=|E|$ edges. The adjacency matrix of $G$ is denoted by $A= \left( A_{i,j}\right)$, where $A_{i,j}= A_{j, i}=1 $\ if $i$\ and $j$\ share an edge and $A_{i, j} = A_{j, i} = 0$ otherwise. 

A modularity maximization problem asks us to identify a community structure $\mathcal{C}=\left\{C_1, C_2,\ldots,C_k \right\}$ of a given graph where each disjoint subsets $C_i$ are called \emph{communities} and $\bigcup_{i=1}^k C_i = V$ so as to maximize the modularity of $\mathcal{C}$. Note that $k$ is not a pre-defined value. The \emph{modularity} \cite{Newman06} of $\mathcal{C}$ is the fraction of the edges that fall within the given communities minus the expected number of such fraction if edges were distributed at random. The randomization of the edges is done so as to preserve the degree of each vertex. If nodes $i$ and $j$ have degrees $d_i$ and $d_j$, then the expected number of edges falling between $i$ and $j$ is $\frac{d_i d_j}{2m}$. Thus, the modularity, denoted $Q$, is then

\begin{eqnarray}
\label{def:modularity}
Q(\mathcal{C}) = \frac{1}{2m}\displaystyle\sum_{i, j} (A_{i,j} - \frac{d_i d_j}{2m})\delta_{ij}
\end{eqnarray}
where $\delta_{ij} = \begin{cases} 1, & \mbox{if } i, j \mbox{ are in the same communities} \\ 0, & \mbox{otherwise}.\end{cases}$.\\
We also define modularity matrix $B$  \cite{Newman06} as 
\[
B_{ij} = A_{ij} - \frac{d_i d_j}{2m}.
\]
%
We note that each row and column of $B$ sum up to zero, hence, $B$ always has the vector $(1, 1, 1,\ldots)$ as one of its eigenvectors.  The same property is also known for the network Laplacian matrix $L=D-A$, where $D$ is diagonal matrix with the $i$th entry to be $d_i$. Laplacian matrix $L$ is widely-used in spectral methods for the graph partitioning that is closely related to our community detection problem. We note that the major difference between the modularity matrix and the Laplacian matrix is that $L$\ is positive-definite while $B$\ is indefinite. As a consequence, while approximation algorithms for the graph partitioning problem  using Laplacian matrix $L$ are available, it is not known if such algorithms are possible for the modularity maximization problem.

 \vspace{-0.02in}
\section{Linear Programming Based Algorithm}
\label{sec:lp}
\subsection{The Linear Program and The Rounding}
The modularity maximization problem can be formulated as an Integer Linear Programming (ILP). The linear program has one variable $d_{i, j}$ for each pair $(i, j)$ of vertices to represent the ``distance'' between $i$\ and $j$ i.e. 
\[ 
d_{i,j} = \left\{
\begin{array}{ll} 
0 & \textrm{ if } i \textrm{ and } j \textrm{ are in the same community}\\
1 & \textrm{ otherwise}.
\end{array}
\right.
\] 
In other words, $d_{i,j}$ is equivalent to $1 - \delta_{i,j}$ in the definition  (\ref{def:modularity}) of modularity. Thus, the objective function to be maximized can be written as $\displaystyle\sum_{i, j}  B_{i, j} (1-d_{i,j})$. We note that there should be no confusion between $d_{i,j}$ the variable representing the distance between vertices $i$ and $j$ and constant $d_i$ (or $d_j$), the degree of node $i$ (or $j$).
The ILP to maximize modularity (IP$_\mathrm{complete}$) is as follows
\begin{align}   
\mbox{maximize }\ \ \, \label{IP:full:obj}   & \quad \frac{1}{2m}\displaystyle\sum_{i, j}  B_{i, j} (1-d_{i,j})&& \\    
\mbox{subject to \quad}     
&                \label{IP:full:triangle1}\,\quad d_{i,j} + d_{j, k} - d_{i,k} \geq 0,&&   \forall i < j < k \\
&                \label{IP:full:triangle2}\,\quad d_{i,j} - d_{j, k} + d_{i,k} \geq 0,&&   \forall i < j < k \\
&                \label{IP:full:triangle3}-d_{i,j} + d_{j, k} + d_{i,k} \geq 0,&&   \forall i < j < k \\   
 &      \label{IP:full:range}  \,\quad d_{i,j} \in [0, 1],&&  i, j \in [1..n], \
\end{align}
Constraints (\ref{IP:full:triangle1}), (\ref{IP:full:triangle2}), and (\ref{IP:full:triangle3}) are well-known \emph{triangle inequalities} that guarantee the values of $d_{i,j}$ are consistent to each other. They imply the following transitivity: if $i$ and $j$ are in the same community and $j$ and $k$ are in the same community, then so are $i$ and $k$. By definition, $d_{i,i} = 0\ \forall i$ and can be removed from the ILP for simplification.

To avoid solving ILP, that is also NP-hard, we  instead solve the LP relaxation of the ILP, obtained by replacing the constraints $d_{i,j} \in \{0, 1\}$ by $d_{i, j} \in [0, 1]$. We shall refer to the IP described above as IP$_\mathrm{complete}$ and its relaxation as LP$_\mathrm{complete}$. If the optimal solution of this relaxation is an integral solution, which is very often the case \cite{Cafieri10}, we have a partition with the maximum modularity. Otherwise, we resort on rounding the fractional solution and use the value of the objective as an upper-bound that enables us to lower-bound the gap between the rounded solution and the optimal integral solution. 

G. Agarwal and D. Kempe \cite{Agarwal08} use a simple rounding algorithm proposed by Charikar et al. \cite{Charikar04} for the \emph{correlation clustering} problem \cite{Bansal02}. The values of $d_{i,j}$ are interpreted as a metric ``distance'' between vertices. The algorithm repeatedly groups all vertices that are close by to a vertex  into a community. The final community structure are then refined by a Kernighan-Lin \cite{Kemighan70} based local search method.  

Since the rounding phase is comparatively simple, the burden of both time and memory comes from solving the large LP relaxation. The LP has ${ n \choose 2}$ variables and $3 {n \choose 3} = \theta(n^3)$ constraints that is about half a million constraints for a network of $100$ vertices, thereby limiting the the size of networks to few hundred nodes. Thus, there is a need to achieve the same guarantees with smaller resource requirements. By combining mathematical approach with combinatorial techniques, we achieve this goal in next subsection.         
\subsection{The Sparse Metric}
In this subsection, we devise an improved LP formulation for the modularity maximization  problem with much fewer number of constraints while getting the same guarantees on the performance.  

Instead of using $3{n \choose 3}$ triangle inequalities to ensure that $d_{i,j}$ is a metric (or pseudo-metric as defined later), we show that only a compact subset of inequalities, so-called \emph{sparse metric}, are sufficient to obtain the same fractional optimal solution.

A function $d$ is a pseudo-metric if $d(i, j) = d_{i,j}$ satisfy the following conditions:
\begin{enumerate}
   \item $d(i, j) \geq 0$\:\qquad \quad  (non-negativity)
   \item $d(i, i) = 0$ \qquad (and possibly $d(i, j) = 0$\ for some distinct values $i\ne j$)
   \item $d(i, j) = d(j, i)$ \quad     (symmetry)
  \item $d(i, j) \leq d(i, k) + d(k, j)$\quad (transitivity).
\end{enumerate}
It is clear that $d$ is an feasible solution of $LP_\mathrm{complete}$ if and only if $d$ is a pseudo-metric within the interval $[0, 1]$. 

Our new linear programming with the \emph{Sparse Metric} technique, denoted by IP$_\mathrm{sparse}$, is as follows: 
\begin{align}
\mbox{maximize }\ \ \, \label{IP:sparse:obj}   & \quad -\frac{1}{2m}\displaystyle\sum_{i, j}  B_{i, j} d_{i,j}&& \\    
\mbox{subject to \quad}     
& \label{IP:sparse:triangle1}\,\quad d_{i,k} + d_{k, j} \geq d_{i,j} &&   k \in N(i,j) \\
 &      \label{IP:sparse:range}  \,\quad d_{i,j} \in \{0, 1\},&&   
\end{align}
The objective can be simplified to $-\frac{1}{2m}\displaystyle\sum_{i, j}  B_{i, j} d_{i,j}$ since $\displaystyle\sum_{i,j}B_{i,j} = 0$. Let $N(i)$ and $N(j)$ denote the set of neighbors of $i$ and $j$, respectively. The set $N (i, j)$ is defined as the union of neighbors of $i$ and $j$
\[
N (i, j) = N(i) \cup N(j) - \{i, j\}
\]
Therefore, the total number of constraints in the formula is upper bounded by 
\begin{align*}
 & \sum_{i < j} d_i + d_j
= (n-1)\sum_{i=1}^{n}  d_i= O(mn)
\end{align*}
When the considered network is sparse, which is often true for complex networks, our new formulation substantially reduces time and memory requirements. For most real-world network instances, where $n \approx m$, the number of constraints is effectively reduced from $\theta(n^3)$ to $O(n^2)$. If we consider the time to solve linear programming to be cubic time the number of constraints, the total time complexity for sparse networks improves to $O(n^6)$ instead of $O(n^9)$ as in the original approach. In practice, LPs can be solved quite efficiently. We mention the increase of the size of the largest solved instance of traveling salesman problem from 49 cities in 1954 \cite{Dantzig54} to 85,9000 cities in 2009 \cite{Applegate09} as an example of rapid development of mathematical programming solvers and computer powers.

Again, we can obtain the relaxation of IP$_\mathrm{sparse}$, described in (\ref{IP:sparse:obj}) to (\ref{IP:sparse:range}), by replacing the constraints $d_{i,j} \in \{ 0, 1\}$ by $d_{i, j} \in [0, 1]$. We shall refer to this relaxation of IP$_\mathrm{sparse}$ as LP$_\mathrm{sparse}$. The fractional optimal solution of this relaxation can also be rounded and tuned with the same algorithms in the previous subsection.
\vspace{-0.05in}
\subsection{Correctness and Performance Guarantees}
In order to achieve the same guarantees provided by solving LP$_\mathrm{complete}$, we show the equivalence of the sparse formulation and the complete formulation:
\begin{itemize}  
  \item IP$_\mathrm{sparse}$ and IP$_\mathrm{complete}$ share the same set of optimal integral solutions (Theorem \ref{theo:ip_eq}).
  \item The optimal fractional solutions of LP$_\mathrm{sparse}$ and LP$_\mathrm{complete}$ have same objective values (Theorem \ref{theo:lp_eq}) i.e. they provide the same upper bound on the maximum possible modularity. 
\end{itemize} 
 Hence, solving LP$_\mathrm{sparse}$ indeed gives us an optimal solution of  LP$_\mathrm{complete}$, while doing so significantly reduces the time and memory requirements. 
 

\begin{theorem}
Two integer programmings IP$_\mathrm{sparse}$ and IP$_\mathrm{complete}$ share the same set of optimal solutions.
\label{theo:ip_eq} 
\end{theorem}
\begin{IEEEproof}
We need to show that every optimal solution of IP$_\mathrm{complete}$ is also a solution of IP$_\mathrm{sparse}$ and vice versa.
 
In one direction, since the constraints in IP$_\mathrm{sparse}$ is a subset of constraints in IP$_\mathrm{complete}$, every optimal solution of IP$_\mathrm{complete}$ will also be a solution of IP$_\mathrm{sparse}$.  

In the other direction, let $d_{i, j}$ be an optimal integral solution of IP$_\mathrm{sparse}$. We shall prove that $d_{i,j}$ must be a pseudo-metric that implies $d_{i,j}$ is also a feasible solution of IP$_\mathrm{complete}$.

For convenience, we assume that the original graph $G=(V, E)$ has no isolate vertices that were known to have no affects on modularity maximization \cite{Newman06}. Construct a graph $G_d = (V, E_d)$ in which there is an edge $(i, j)$ for every $d_{i, j}=0$. Let $\mathcal{C}_d = \{ C_d^1, C_d^2, \ldots,C_d^l \}$ be the set of connected components in $G_d$, where $C_d^t$ represents the set of vertices in $t$th connected components.
\begin{proposition}
Every connected component $C_d^i$ induces a \emph{connected} subgraph in $G=(V,E)$.
\label{prop:con}
\end{proposition}
\begin{IEEEproof}
We prove by contradiction. Assume that the connected component $C_d^t$ does not induce a connected subgraph in $G$. Hence, we can partition $C_d^t$ into two subsets $S$ and $T$ so that there are no edges between $S$ and $T$ in $G$.

Construct a new solution $d'$ from $d$ by setting $d'_{i, j}=1$ for all pairs $(i, j) \in P(S, T)$, the set of \emph{pairs} with one end point in $S$ and one endpoint in $T$. Since, $A_{i, j} = 0\ \forall (i, j) \in P(S,T)$, we have $B_{i, j} = A_{i, j} - \frac{d_i d_j}{2m} < 0\ \forall (i, j) \in P(S, T)$. Hence, setting $d'_{i, j} = 1\ \forall (i, j) \in P(S,T)$ can only increase the objective value. In fact, doing so will strictly increase the objective. There must be at least one pair $(i, j) \in P(S,T)$ with $d_{i,j}=0$, or else $C_d^i$ is not a connected component in $G_d$. 

It is not hard to verify that $d'_{i, j}$ satisfy all constraints of IP$_\mathrm{sparse}$ since those triangle inequalities must involve at least one edge in the original graph $G$, while $S$ and $T$ are disconnected sets in $G$.

Thus, we have derived from an optimal solution a new feasible solution with higher objective (contradiction).
\end{IEEEproof} 

\begin{figure}
\centering
\includegraphics[width=0.26 \textwidth]{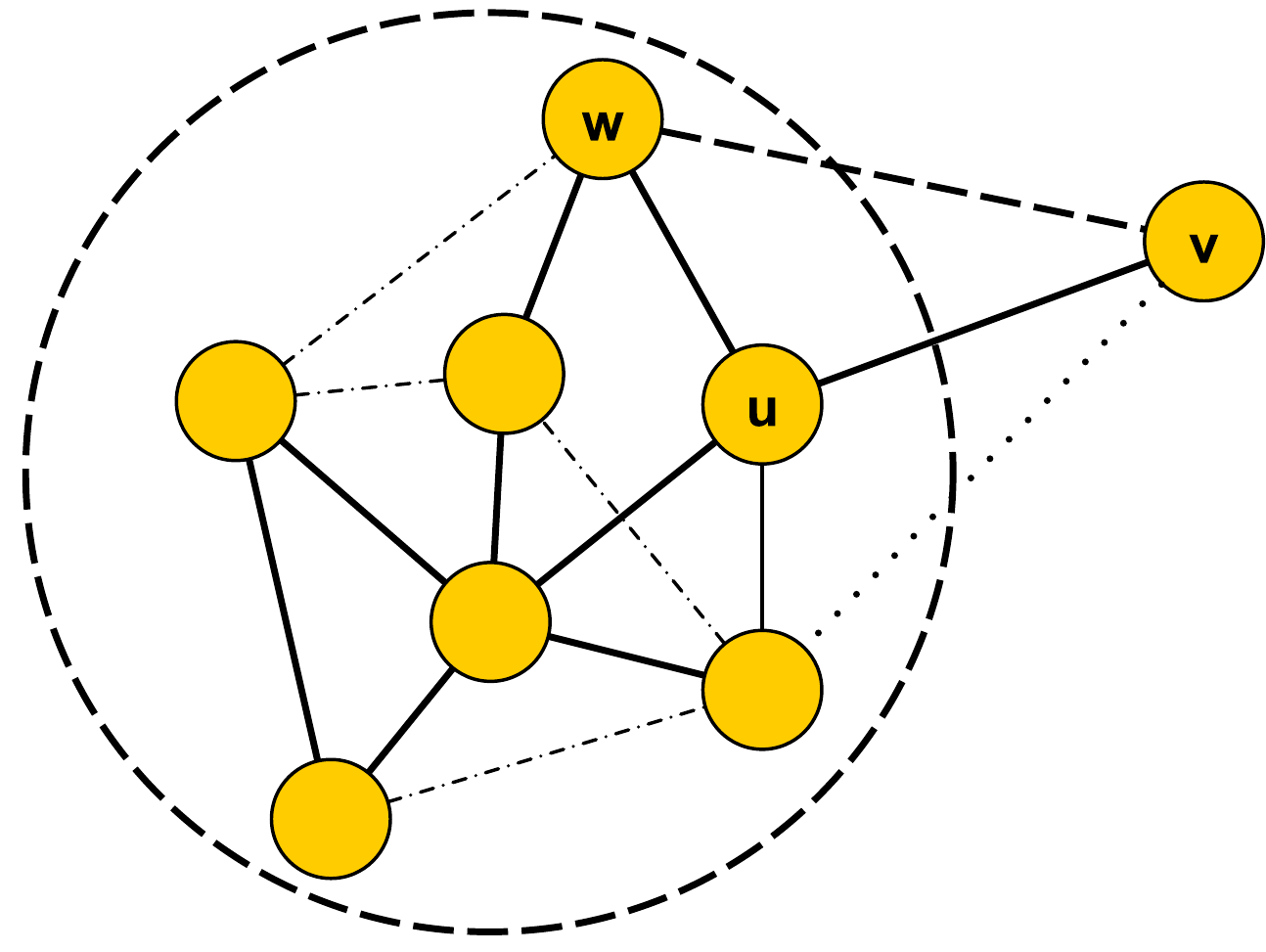}
\caption{Clique expanding process.}
\label{fig:clique}
\vspace{-0.15in}
\end{figure}

The rest is to prove that for each connected component $C_d^t$ of $G_d$,  if $i, j \in C_d^t$\ then the distance $d_{i, j} = 0$.  We prove by repeatedly applying a ``clique expanding'' process. At each step, every pair of vertices in the clique are proven to have distance $0$. Then, we expand the clique, adding one more adjacent vertex to the clique and prove that the new clique also has vertices of distance zero from each other (see Fig. \ref{fig:clique}). 

\emph{Initial step}. We first prove there is an edge $(i, j) \in E$ of the original graph $G$ satisfying $d_{i, j} = 0$. We shall choose that edge as our initial clique of size 2. Assume no such edge exists, all pairs $d_{i, j} = 0$ within $C_d^t$ have $A_{i,j}=0$ and $B_{i,j} < 0$. Thus, again we  can increase the distance of all pairs with $d_{i,j}=0$ to $1$ without violating any constraints, while increasing the objective value (contradiction). Therefore, we can always find an edge that belongs to  both $G$ and $G_d$.

\emph{Expanding steps}. Denote our clique by $K_t$. If $K_t = C_d^t$, then we can complete the proof for $C_d^t$. Otherwise, there is a vertex $u \in K_t$ and a vertex $v \in C_d^t-K_t$, so that $(u, v)$ is an edge in both $G$ and $G_d$ ($d_{u, v}=0$). The existence of such an edge $(u, v)$ can be proven by contradiction (Assume not, then increase distance of all pairs in $P(K_t, C_d^t - K_t)$ from $0$ to $1$ to increase the objective value while not violating any constraints.). Then, for each vertex $w \in K_t - \{ u \}$, the constraint $d_{w, u} + d_{u, v} \geq d_{w, v}$ is in IP$_\mathrm{sparse}$ and $d_{w, u} =0$ from the property of $K_t$. It follows that $d_{w, v} = 0$ for all $w \in K_t$. By adding $v$ to $K_t$ we increase the size of the clique, while ensuring the zero-distance property.
 
Since the size of $C_d^t$ is at most $n$, the expanding process will finally terminate with $K_t = C_d^t$.
\end{IEEEproof}

\begin{theorem}
\label{theo:lp_eq}
  LP$_\mathrm{sparse}$ and LP$_\mathrm{complete}$ share the same set of fractional optimal solutions.
\end{theorem}  
\begin{IEEEproof}
We need to show that every \emph{fractional optimal} solution of LP$_\mathrm{complete}$ is also a \emph{fractional} solution of LP$_\mathrm{sparse}$ and vice versa. Since the integrality constraints have been dropped in both LP relaxations, we need a different approach to the proof in Theorem \ref{theo:ip_eq}. 
 
One direction is easy, every fractional optimal solution of  LP$_\mathrm{complete}$ is also a \emph{fractional} solution of LP$_\mathrm{sparse}$.
  
For the other direction, let $d_{i,j }$ be a fractional optimal solution of LP$_\mathrm{sparse}$, we shall prove that $d_{i, j}$ is also a feasible solution of LP$_\mathrm{complete}$.
 
Associate a weight $w_{i,j}=d_{i, j}$ for each edge $(i, j) \in E$ (other edges are assigned weights $\infty$). Let $d'_{i, j}$ be the distance between two nodes $(i, j)$ with the new edge weights. We have
\begin{enumerate}
  \item $d'_{i, j} \geq d_{i, j}$ for all $i, j$ and $d'_{i, j} = d_{i, j} \forall (i, j) \in E$. 
  \item $d'_{i, j} = \min_{k=1}^{n} \{ d'_{i, k} + d'_{k, j} \}$. Hence, $d'_{i,j}$ is a pseudo-metric.
\end{enumerate}
The first statement can be shown by applying the triangle inequalities in LP$_\mathrm{sparse}$.  Since, $d'_{i,j}$ be the shortest distance between $i$ and $j$ in $G$, there is a path $u_0 = i, u_1, \ldots,u_l = j$ with the length $d'_{i, j} = d_{u_0, u_1} + d_{u_1, u_2}+\ldots+d_{u_{l-1}, u_l}$. Since $(u_{k-1}, u_k)$ are edges in $G$ for all $k=1..l$, we can apply triangle inequalities iteratively  
\begin{align}
\nonumber
d_{i, j} &\leq d_{u_0, u_1} + d_{u_1, u_l} \leq d_{u_0, u_1} + d_{u_1, u_2}+d_{u_2, u_l}\\
&\leq \ldots \leq d_{u_0, u_1} + d_{u_1, u_2}+\ldots+d_{u_{l-1}, u_l} = d'_{i, j}
\end{align}
If $(i,j) \in E$, we have $d'_{i, j} \leq d_{i, j}$. Hence, $d'_{i,j} = d_{i, j }\ \forall (i,j) \in E$.
The second statement comes from the definition of $d'_{i, j}$. 

Notice that  $d'_{i, j}$ may be no longer upper bounded by one. Therefore, we define $d^*_{i, j}=\min\{ d'_{i, j}, 1\}$. We also have
\[
  d^*_{i, j} \geq d_{i, j}\ \forall i, j \mbox{ and }  d^*_{i, j} = d_{i, j}\ \forall (i, j) \in E.
\] 
And more importantly,  $d^*$ is also a pseudo-metric. Since $d^*_{i, k} + d^*_{k, j} \geq \min\{d'_{i, k} + d'_{k, j}, 1 \} \geq \min \{ d'_{i, j}, 1\} = d^*_{i, j}$.

Now, if $d_{i, j} = d^*_{i, j}$ for all $i, j$, then $d$ satisfies all triangle inequalities in LP$_\mathrm{complete}$ and we yield the proof.

Otherwise, assume that $d_{i, j} < d^*_{i, j}$ for some pair $(i, j)$. We show that $d^*$ is a feasible solution of LP$_\mathrm{sparse}$ with greater objective value that contradicts the hypothesis that $d$ is an optimal solution.

Since for all edges $(i, j) \notin E$, $d_{i,j} = d^*_{i, j}$, and for pairs $(i, j) \notin E$, $B_{i,j} < 0$\ and $d^*_{i, j} \geq d_{i,j}$, we have
 $\sum_{i,j} d^*_{i, j} > \sum_{i, j}d_{i, j}$ (contradiction). 
\end{IEEEproof}

\vspace{-0.01in}
\section{Approximation Algorithms for Maximizing Modularity in Power-law Networks}
\label{sec:approx}
This section presents approximation algorithms for the modularity maximization problem in power-law networks. A factor $\rho$ approximation algorithm for a maximization problem, find in polynomial- time a solution with the value no less than $\rho$ times the value of an optimal solution. Approximation algorithms are  being used for problems where exact polynomial-time algorithms are too expensive and in many cases, they can yield valuable insights to the problem.

We make a detour to focus on the problem of modularity maximization in division of the network into just two communities. The maximum modularity value of the division into two communities are shown to ``close'' to the best possible modularity. Thus, an approximation algorithm for the division into two communities problem also yields an approximation algorithm for the modularity maximization problem.
\vspace{-0.01in}
\subsection{Division into $k$ Communities}
Let  $Q_k$\ be the maximal modularity obtained by a division of the network into \textbf{exact} $k$\ communities. We also denote $Q_k^+= \max_{i=1}^k  Q_i$\ and $Q_\mathrm{opt} = Q_n^+$, the best possible modularity over all possible divisions. Let $\delta^\mathrm{opt}$\ be a community structure with the maximum modularity $Q_\mathrm{opt}$. 

\begin{proposition}
$Q_{1}=0$\ and  $Q_{n} = -\frac{ \sum_i {d_i^2} } {4m^2}$.
\end{proposition}
\begin{lemma}
\label{lem:Q2}
\[
Q_k^+ \geq (1-\frac{1}{k}) Q_\mathrm{opt}
\]
\end{lemma}
\vspace{-0.05in}
\begin{proof}
If $\delta^\mathrm{opt}$ has at most $k$ communities, than we  have $Q_{k}^+ = Q_\mathrm{opt}$. Otherwise $\delta^\mathrm{opt}$ has more than $k$ communities.

We can rewrite the modularity as
\[
Q_\mathrm{opt} = \frac{1}{2m} \sum_{\delta^{opt}_{ij}=1}{B_{ij}}
\]

\vspace{-0.05in}
Construct a $k$-division of the network by randomly assigning communities in $\delta^\mathrm{opt}$ into one of $k$ new ``\emph{super}'' communities. Let $\delta^k$ denote the obtained partitioning.
If $\delta^\mathrm{opt}_{ij}=1$, then $\delta^k_{ij}=1$ i.e. all within intra-communities pairs remain within new ``super'' communities. All pairs $(i, j)$ with $\delta^\mathrm{opt}_{ij}=0$ (inter-community pairs) become intra-communities pairs 
with probability $1/k$. Hence, the contribution of a pair $(i, j)$\ with $\delta^\mathrm{opt}_{ij}=0$ to the expected modularity is $\frac{1}{k} B_{ij}$. Hence, the expected modularity of the $k$-division by randomly grouping communities will be
\begin{align*}
Q_{E}\ &=\quad \frac{1}{2m} \bigg(\sum_{\delta^{opt}_{ij}=1}{B_{i,j}}+ \frac{1}{k}\sum_{\delta^{opt}_{i,j}=0}{B_{i,j}} \bigg) \\
          &=\quad \frac{1}{2m} \left(1-\frac{1}{k}\right)\sum_{\delta^{opt}_{i,j}=1}{B_{i,j}} 
           = \left(1-\frac{1}{k}\right) Q_\mathrm{opt}                  
\end{align*}
In the second step, we have used the equality  $\sum_{ij} B_{i,j} = 0$ or
 equivalently $\sum_{\delta^\mathrm{opt}_{i,j}=1}B_{i,j} =
  -\sum_{\delta^\mathrm{opt}_{i,j}}B_{i,j}$. Therefore, we have $Q_k^+ \geq Q_{E} =\left( 1-\frac{1}{k}\right){Q_\mathrm{opt}}$.
\end{proof}

It follows from Lemma \ref{lem:Q2} that an approximation algorithm with a factor $\rho$ for maximizing $Q_2$\ will also be an approximation with a factor $2\rho$ to the modularity  maximization problem.

For a division of the network into two groups define
 \[
 x_{i} = \begin{cases} 1, & \mbox{if } i \mbox{ belong to community } 1 \\ 
       -1, & \mbox{if } i \mbox{ belong to community } 2.\end{cases}
 \]

We can write the modularity for the division into two communities as 
\[
Q = \frac{1}{4m}\displaystyle\sum_{i, j} B_{i,j}(x_i x_j + 1)             \\
        = \frac{1}{4m}\displaystyle\sum_{i, j} B_{i,j}x_i x_j \\
        = \frac{1}{4m} x^{\mathrm{T}} B x
\] 
Hence, the division into two communities is a special case of the maximizing quadratic program problem i.e. the problem of finding
a vector $x \in \{-1, 1\}^n$\ such that  $x^{\mathrm{T}} B x$ is maximized. The following results was due to M. Charikar et al. \cite{Charikar04} and Nesterove et al. \cite{Nesterove97}.
\begin{theorem}\cite{Charikar04}
\label{theo:quad}
Given an arbitrary matrix $A$, whose diagonal elements are nonnegative, the problem of finding
$x \in \{-1, 1\}^n$\ such that  $x^{\mathrm{T}} B x$\ is maximized can be approximated within $O(\log n)$. In case $B$\ is positive definite, the ratio can be improved to $\frac{\pi}{2}$\ \cite{Nesterove97}.
\end{theorem}
Unfortunately, the matrix $B$ is not positive definite. Even worse, the main diagonal contains all negative entries as the $i$th entry is $-\frac{d_i^2}{4m^2}$. Hence, we cannot directly apply above results for the division into two communities problem.
\vspace{-0.1in}
\subsection{Power-law Networks}
Complex networks including social, biological, and technology networks display a non-trivial topological feature: their degree sequences can be well-approximated by a power-law distribution \cite{Girvan02}. At the same time they exhibit modular property i.e. the existence of naturally division into communities. We establish the connection between the power-law degree distribution property and the modular property, stating that whenever a network have power-law degree distribution, there is presence of communities in the network with a significant modularity.  


We use the well-known  $P(\alpha, \beta)$\ model by  F. Chung and L. Lu \cite{Aiello:massive} for power-law networks  in which there are $y$\ vertices of degree $x$, where $x$\ and $y$\ satisfy
$
\log y = \alpha - \beta \log x
$. In other words,
\[ |\{ v: d(v) = x\}| = y = \frac{e^\alpha}{x^\beta} \]

Basically, $\alpha$\ is the logarithm of the size of the graph ($n=e^\alpha$) and $\beta$\ is the log-log growth rate of the graph. While the scale of the network depends on $\alpha$,  $\beta$ decides the connection pattern  and many other important characterizations of the network. Different networks at different scales with same $\beta$ often exhibit same characteristics. For instance, the larger $\beta$, the sparser and the more ``power-law'' the network is. Hence, $\beta$ is regarded as a constant in $P(\alpha, \beta)$ model.

In $P(\alpha,\beta)$ model, the maximum degree in a $P(\alpha, \beta)$\ graph is $e^\frac{\alpha}{\beta}$. The number of vertices and edges are
\noindent\begin{align}
\nonumber n =\displaystyle\sum_{x=1}^{e^{\frac{\alpha}{\beta}}}\frac{e^\alpha}{x^\beta} \approx \left\{
\begin{array}{ll}
 \zeta(\beta)e^\alpha & \mbox{if } \beta>1\\
 \alpha e^\alpha & \mbox{if } \beta =1 \\
 \frac{e^{\frac{\alpha}{\beta}}}{1-\beta} & \mbox{if } \beta < 1
\end{array}
\right.
,\\ 
\label{eq:powerlaw}
 m = \frac{1}{2}\displaystyle\sum_{x=1}^{e^{\frac{\alpha}{\beta}}}x\frac{e^\alpha}{x^\beta} \approx \left\{
\begin{array}{ll}
 \frac{1}{2}\zeta(\beta-1)e^\alpha & \mbox{if } \beta >2\\
 \frac{1}{4}\alpha e^\alpha & \mbox{if } \beta = 2\\
 \frac{1}{2}\frac{e^{\frac{2\alpha}{\beta} } }{2-\beta} & \mbox{if }  \beta < 2
\end{array}
\right.
\end{align}
where 
$\zeta(\beta) = \sum_{i=1}^{\infty}{\frac{1}{i^\beta}}$\ is the Riemann Zeta function. Without affecting the conclusions, we will simply use real number instead of rounding down to integers. The error terms can be easily bounded
and are sufficiently small in our proofs.

\begin{figure}
\centering
\subfloat[Following algorithm]{
\includegraphics[width=0.24 \textwidth]{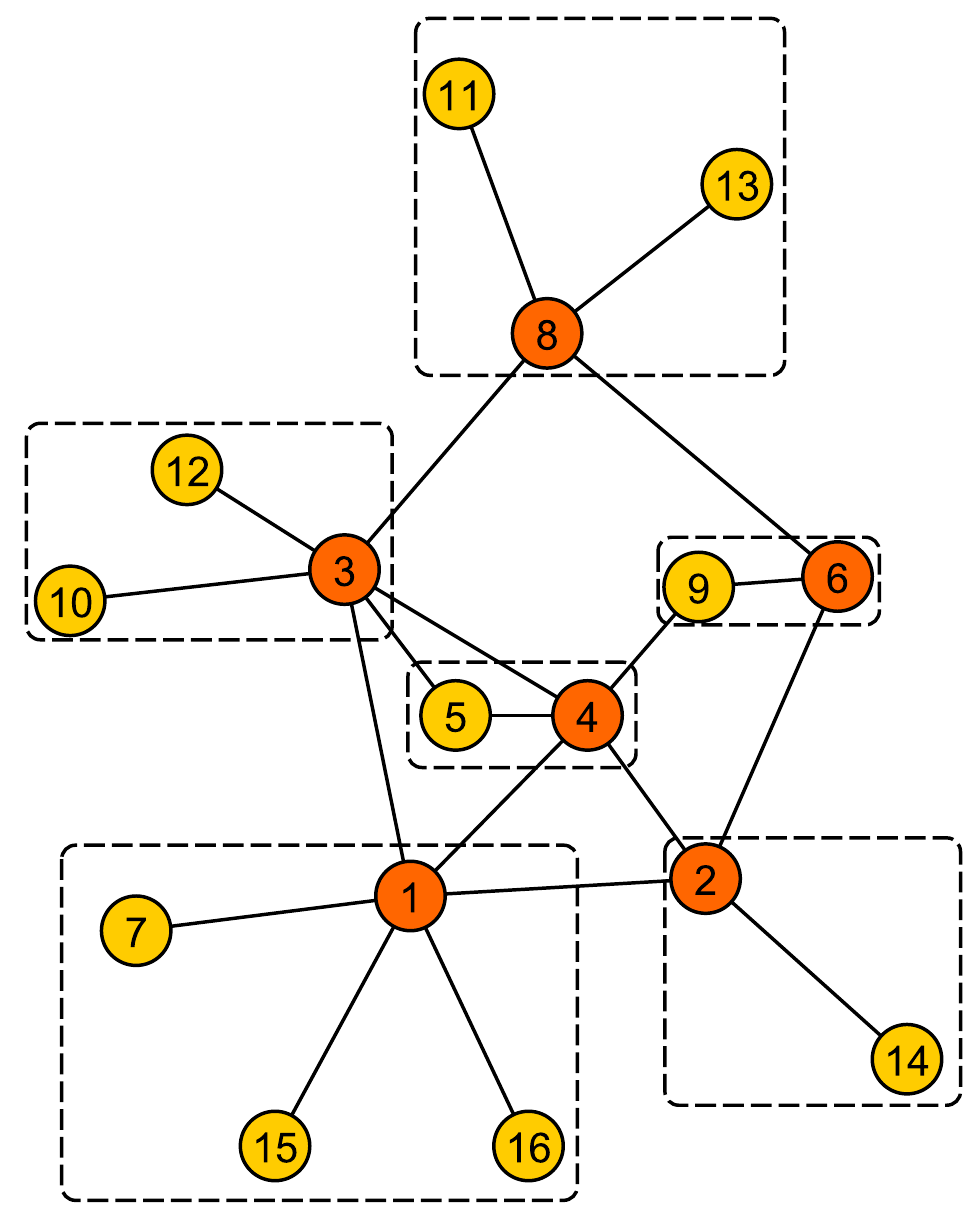}
\label{fig:forest} 
}
\subfloat[Optimal community structure] {
\includegraphics[width=0.24 \textwidth]{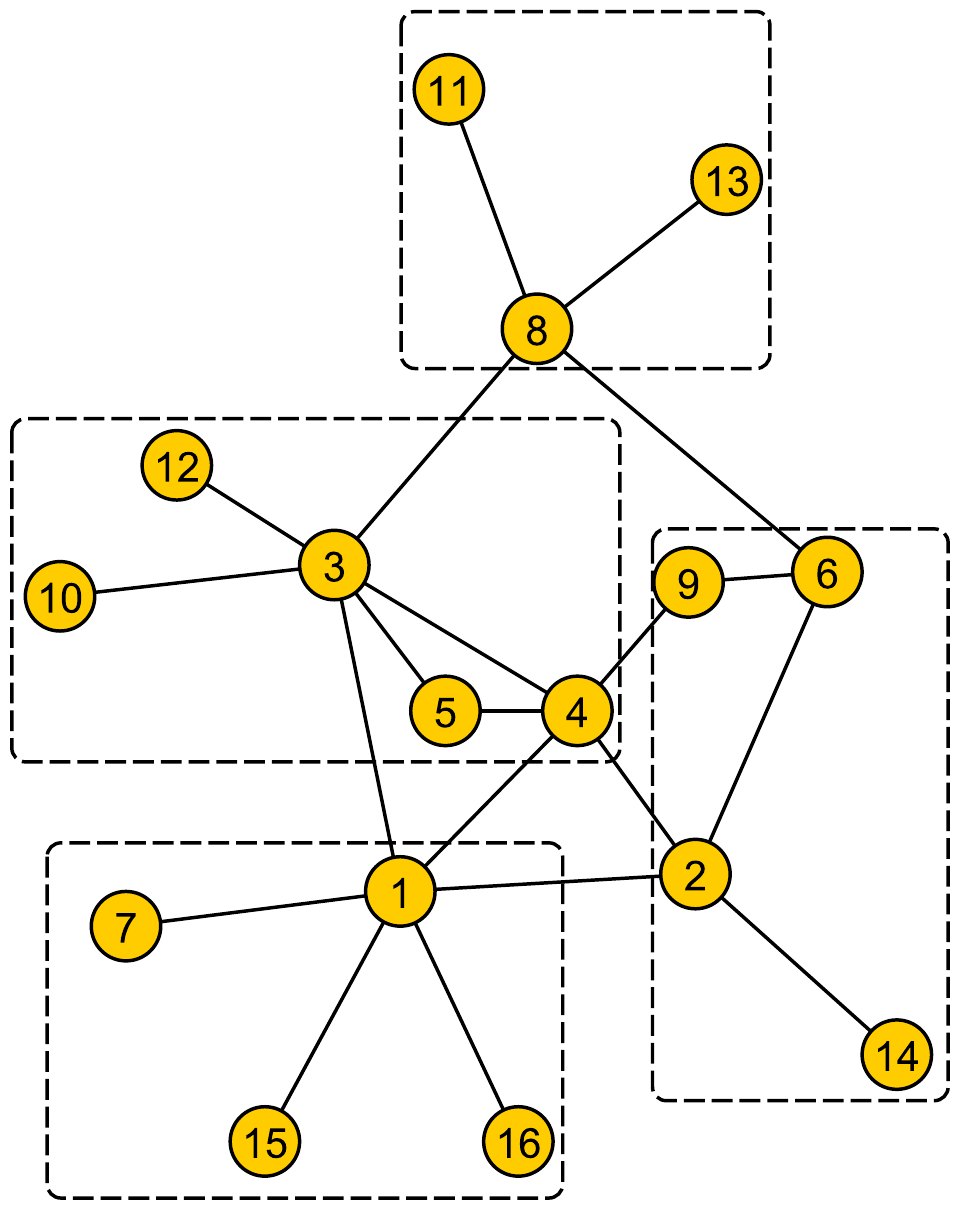}
\label{fig:forest_opt}
}
\caption{On the left, a community structure found by Following Algorithm in Theorem \ref{theo:approx} when $d_0=2$. Each rounded square represents a community. and followees are in the darker color. The modularity is 0.325 i.e. 87\% of the optimal modularity, 0.374. On the right, the optimal community structure found by solving IP$_\mathrm{sparse}$.}
\vspace{-0.15in}
\end{figure}
Most real-world networks have the log-log growth rate $\beta$ between $2$ and $3$. For examples, scientific collaboration networks with $2.1 < \beta <2.45$ \cite{Barabasi02}, Word Wide Web with $\beta$ for in-degree and out-degree of $2.1$ and $2.45$, respectively \cite{Albert00}; Internet at router and intra-domain level with $\beta=2.48$ and so on. No power-law networks with $\beta <1$ have been observed. One of the reason is that when $\beta < 1$, the number of edges $m = \Omega(n^2)$ i.e.  the network is not ``scale-free''.

\begin{theorem}
\label{theo:approx}
There is an $O(\log n)$ approximation algorithm for the  modularity maximization problem in 
power-law networks with the log-log growth rate $\beta > 1$. If $\beta > 2$, the problem can be approximated within a constant approximation factor $2\zeta(\beta-1)$, where $\zeta(x) = \sum_{i=1}^{\infty}{\frac{1}{i^x}}$\ is the Riemann Zeta function.
\end{theorem}
\begin{proof}
From Lemma (\ref{lem:Q2}) with $k=2$, we have $\frac{1}{2} Q_\mathrm{opt} \leq Q_2^+$. Hence, it is sufficient to approximate $Q_2^+$\ within a factor of $O(\log n)$.

We have
\begin{align}
\label{eq:diagonal}
\nonumber Q_2^+ &=\frac{1}{4m}\ \displaystyle \max_{x \in \left\{-1, 1 \right\}^n } x^T B x \\
&=\frac{1}{4m}\  \displaystyle \max_{x \in \left\{-1, 1 \right\}^n } x^T B_0 x  - \sum_{i=1}^{n} \frac{d_i^2}{8m^2},
\end{align}
where $B_0$ is obtained by replacing the diagonal of $B$\ with zeros.

\begin{table}[tbp]
  \centering
  \caption{Order and size of network instances}
    \begin{tabular}{ccrr}
    \addlinespace
    \toprule
    Problem ID & Name  & Nodes n & Edges m \\
    \midrule
    1     & Zachary's karate club & 34    & 78 \\
    2     & Dolphin's social network      &  62      &  159\\
    3     & Les Miserables      &    77   &  254 \\
    4     & Books about US politics      &  105    & 441  \\
    5     & American College Football     &    115  & 613  \\
    6     & US Airport 97      &  332    &   2126\\
    7     & Electronic Circuit (s838)      &   512   &  819  \\
    8     & Scientific Collaboration      & 1589     & 2742   \\            
    \bottomrule
    \end{tabular}%
  \label{tab:sum}%
  \vspace{-0.2in}
\end{table}%

Let $D = \sum_{i=1}^{n} \frac{d_i^2}{8m^2}$, the second term in equation (\ref{eq:diagonal}). We can approximate 
\[ 
\mathrm{OPT}_0=\displaystyle \max_{x \in \left\{-1, 1 \right\}^n } x^T B_0 x = Q^+_2 + D
\]
within a factor of $O(\log n)$\ by the method in Theorem \ref{theo:quad}. That means we can find a division of the network into two communities with the modularity is at least 
\begin{align*}
&\frac{c}{\log n} \mathrm{OPT}_0 - D = \frac{c}{\log n} (Q^+_2 + D) - D\\
\geq \quad  &\frac{c}{\log n} Q^+_2 - D \geq \frac{c}{2\log n} Q_\mathrm{opt} - D
\end{align*}
where $c$ is an independent constant.

If we can show that $D  = o\left( \frac{1}{\log n}\mathrm{OPT}_0 \right)$, then we can approximate the maximum modularity within  a factor $O(\log n)$. This is equivalent to 
\begin{align}
\label{eq:lim}
\displaystyle\lim_{n \rightarrow \infty} \frac{Q_\mathrm{opt}}{D \log n} = \infty \mbox{ or }
\displaystyle\lim_{\alpha \rightarrow \infty} \frac{Q_\mathrm{opt}}{D \log n} = \infty
\end{align}

To show (\ref{eq:lim}), we present a linear-time algorithm, called \emph{Following}, to find a community structure $\mathcal{L}$ with a lower bound on the modularity. An illustration example for the algorithm is shown in Fig. \ref{fig:forest}.

\vspace{0.1in}
\noindent\framebox[0.49\textwidth]{  
\centering  
\begin{minipage}{.46\textwidth}  
\textbf{Following Algorithm} ( Parameter $d_0 \in \mathbb{N}^+$)
\begin{enumerate}[i.]
\item Start with all nodes unlabeled
\item Sort nodes in \emph{non-decreasing order of degree}
\item For each unlabeled node $v$ with $d_v\leq d_0$, find a neighbor $u$ that is not a follower; set $v$ to follow $u$ i.e. label $v$ ``follower'' and $u$ ``followee''. If many such $u$ exist, select the one with the minimum degree.
\item Label all unlabeled nodes ``followee''.
\item Put each followee and its followers into a community.
\end{enumerate}
\end{minipage}
} 

\vspace{0.05in}

Despite that higher values of $d_0$  possibly lead to better approximation ratios, it is sufficient for our proof to consider only the case $d_0=1$. That means all leaf nodes will attach to (follow) their neighbors. Assume that for a graph $G=(V,E)$, vertices in $V$ are numbered so that leaf nodes will have higher numbering than non-leaf nodes i.e. $V = \{v_1, v_2,\ldots,v_t, \underbrace{ v_{t+1},\ldots,v_n}_{\rm leaf\ nodes} \}$\ in which $t$ is the number of non-leaf nodes. For a node $v_i, i=1\ldots t$, let $l_i \leq d_i$\ be the number of leaves attached to $v_i$. There will be $t$\ communities associated with $v_1, v_2,\ldots,v_t$, respectively. 

Since there are $e^\alpha$ vertices of degree one, there are at least $\frac{1}{2} e^\alpha$\ edges inside considered communities. Hence, 
\begin{align}
\label{eq:qleaves}
\nonumber Q(\mathcal{L}) &= \frac{e^\alpha}{2m} - \sum_{i=1}^{t} \frac{\left(d_i+l_i\right)^2}{4m^2}
\geq \frac{e^\alpha}{2m} - \sum_{i=1}^{n} \frac{4d_i^2}{4m^2}\\
&= \frac{e^\alpha}{2m} - 8D
\end{align}
Since $Q_\mathrm{opt} \geq Q(\mathcal{L})$, instead of showing (\ref{eq:lim}), we can show 
\[
\displaystyle\lim_{\alpha \rightarrow \infty} \frac{Q(\mathcal{L})}{D \log n} = \infty
\Leftrightarrow \displaystyle\lim_{\alpha \rightarrow \infty} \frac{e^\alpha/2m}{D \log n} = \infty
\]

From the power-law degree distribution in (\ref{eq:powerlaw}): 
\begin{align}
\label{eq:D}
D = \displaystyle\sum_{x=1}^{e^{\frac{\alpha}{\beta}}}\frac{e^\alpha}{x^\beta} \frac{x^2}{8 m^2} = \frac{e^\alpha}{8m^2} \displaystyle\sum_{x=1}^{e^{\frac{\alpha}{\beta}}}x^{2-\beta}
\end{align}
Consider all three cases of $\beta$:

\textbf{Case $\beta > 2$}: Since $x^{2-\beta} < 1$, from equation (\ref{eq:powerlaw}) we have
\begin{align}
\nonumber Q(\mathcal{L})&\geq  \frac{e^\alpha}{2m} - 8D
\geq \frac{1}{\zeta(\beta-1)} - \frac{4e^{\frac{\alpha}{\beta}}}{\zeta(\beta-1)^2 e^\alpha}\\
&\geq \frac{1}{2\zeta(\beta-1)}
\end{align} 
Since $Q_\mathrm{opt} \leq 1$, community structure $\mathcal{L}$ approximate the optimum solutions within a constant factor $2\zeta(\beta-1)$.

\textbf{Case $\beta  = 2$}: We have $\log n < 2 \alpha$. Hence, 
\[
D \log n \leq \frac{2e^\alpha}{\alpha^2 e^{2\alpha}} \bigg(\displaystyle\sum_{x=1}^{e^{\frac{\alpha}{\beta}}}1\bigg) 2 \alpha = \frac{4 e^{\alpha/\beta}} {\alpha e^\alpha}  
\]
Thus,
\[
\displaystyle\lim_{\alpha \rightarrow \infty} \frac{e^\alpha/2m}{D \log n} 
\geq
\displaystyle\lim_{\alpha \rightarrow \infty} \frac{e^\alpha}{2e^{\alpha/\beta}}
= \infty
\]
Hence, the modularity maximization problem can be approximated within a factor $O(\log n)$ in this case.

\textbf{Case   $2 > \beta  > 1$}:  
\begin{align}
\nonumber D\log n \leq &\quad \frac{e^\alpha}{8m^2} e^{\frac{\alpha}{\beta}(3-\beta)}
\displaystyle\sum_{x=1}^{e^\frac{\alpha}{\beta}}
\left( \frac{x}{e^\frac{\alpha}{\beta}} \right)^{2-\beta}
 \frac{1}{e^\frac{\alpha}{\beta}} 2 \alpha\\
\nonumber \leq&\quad  \frac{2\alpha e^\alpha}{\frac{2}{(2-\beta)^2} e^\frac{4\alpha}{\beta} } 
e^{\frac{\alpha}{\beta}(3-\beta)}
\displaystyle\int_{0}^{1} x^{2-\beta} \mathrm{d} x \\
\nonumber \leq&\quad  \frac{(2-\beta)^2 }{ e^\frac{\alpha}{\beta} } 
\frac{\alpha}{3-\beta}
\end{align}

Therefore,
\begin{align*}
&\quad\displaystyle\lim_{\alpha \rightarrow \infty} \frac{e^\alpha/2m}{D \log n} \geq  
\displaystyle\lim_{\alpha \rightarrow \infty} \frac{e^\alpha} {2\frac{e^{2\alpha/\beta}}{2-\beta}} \frac{(3-\beta)e^{\alpha/\beta}}{\alpha (2-\beta)^2} \\
&\geq \displaystyle\lim_{\alpha \rightarrow \infty}  \frac{3-\beta}{\alpha(2-\beta)} e^{\alpha(1-\beta^{-1})} = \infty
\end{align*}
Hence, the theorem follows. 
\end{proof}  


\label{sec:exp}
\begin{table}[tbph]
  \centering
  \caption{The modularity obtained by previous published methods GN \cite{Girvan02}, EIG \cite{Newman06}, VP \cite{Agarwal08}, LP$_\mathrm{complete}$\cite{Agarwal08}, our \emph{sparse metric} approach LP$_\mathrm{sparse}$ and the optimal modularity values OPT \cite{Cafieri10}. The optimal modularity for network 8 (as a whole) has not been known before; we compute it by solving our our IP$_\mathrm{sparse}$ within only 15 seconds.}
    \begin{tabular}{ccrrrrrr}
    \addlinespace
    \toprule
    ID & n  & GN & EIG & VP & LP$_\mathrm{complete}$ & LP$_\mathrm{sparse}$  & OPT\\
    \midrule
    1     &  34    & 0.401 & 0.419 & 0.420 & 0.420 &0.420 & 0.420\\
    2     &   62   & 0.520 & - & 0.526 & 0.529 & 0.529& 0.529  \\
    3     &   77   & 0.540 & - & 0.560 & 0.560 & 0.529& 0.529  \\
    4     &   105  & -     & 0.526 & 0.527 & 0.527 & 0.529& 0.529  \\
    5     &   115  & 0.601 & - & 0.605 & 0.605 & 0.605& 0.605 \\
    6     &   332  & -      & -  & -  & -  & 0.368& 0.368  \\    
    7     &  512   & -  & - & - & - & 0.819& 0.819    \\ 
    8     &  1589  &  -    & -  & -  &  -  & 0.955& 0.955  \\       
    \bottomrule
    \end{tabular}%
  \label{tab:mod}%
  \vspace{-0.15in}
\end{table}%
\section{Computational experiments}
\label{sec:exp}
We present experimental results for our linear programming rounding algorithm in Section \ref{sec:lp}. The LP solver is GUROBI 4.5, running on a PC computer with Intel 2.93 Ghz processor and 12 GB of RAM. We evaluate our algorithm on several standard test cases for community structure identification, consisting of real-world networks. The datasets names   together with their sizes are  are listed in Table \ref{tab:sum}. The largest network consists of 1580 vertices and 2742 edges. All references on datasets can be found in  \cite{Agarwal08} and \cite{Cafieri10}.

\begin{table}[tbph]
  \centering
  \caption{Number of constraints in formulations LP$_\mathrm{complete}$ used in papper \cite{Agarwal08} (Constraint$\langle$C$\rangle$) and the computational time (in seconds) (Time$\langle$C$\rangle$)  versus number of constraints in our \emph{sparse metric} formulation LP$_\mathrm{sparse}$ (Constraint$\langle$S$\rangle$) and its computational time(Time$\langle$S$\rangle$). }
    \begin{tabular}{crrrrrr}
    \addlinespace
    \toprule
    ID & n  & Constraint$\langle$C$\rangle$ & Constraint$\langle$S$\rangle$ & Time$\langle$C$\rangle$ & Time$\langle$S$\rangle$\\
    \midrule
    1     &   34   & 17,952 & 1,441&0.21 &0.02 \\
    2     &   62   &  113,460& 5,743&3.85 &0.11 \\
    3     &   77   &  219,450& 6,415&13.43 &0.08 \\
    4     &   105  &  562,380& 30,236&60.40 &1.76 \\
    5     &   115  &  740,715& 66,452&106.27 &13.98 \\
    6     &   332  &  18,297,018& 226,523& - &197.03\\
    7     &   512  &  66,716,160& 294,020& - &53.18 \\
    8     &  1589  &   2,002,263,942& 159,423& - &2.94 \\        
    \bottomrule
    \end{tabular}%
  \label{tab:perm}%
  \vspace{-0.15in}
\end{table}%

Since the same rounding procedure are applied on the optimal fractional solutions, both LP$_\mathrm{complete}$ and LP$_\mathrm{sparse}$ yield the same modularity values. However, LP$_\mathrm{sparse}$ can run on much larger network instances. The modularity of the rounding LP algorithms and other published methods are shown in Table \ref{tab:mod}. The rounding LP algorithm can find optimal solutions ( or within 0.1\% of the optimal solutions) in all cases. The source code for our LP algorithm can be obtained upon request.

Finally, we compare the number of constraints of the LP formulation used in \cite{Agarwal08} and our new formulation (LP$_\mathrm{sparse}$) in Table \ref{tab:perm}. Our new formulation  contains substantially less constraints, thus can be solved more effectively. The old LP formulation cannot be solved within the time allowance (10000 seconds) and the memory availability (12 GB) in cases of the network instances 6 to 8. The largest instance of 1589 nodes is solved surprisingly fast, taking under 3 seconds. The reason is due to the presence of leaves (nodes of degree one) and other special motifs that can be efficiently preprocessed with the reduction techniques in \cite{Arenas07}.

Our new technique substantially reduces the time and memory requirements  both theoretically and experimentally without any trade-off on the quality of the solution. The size of solved network instances raises from hundred to several thousand nodes while the running time on the medium-instances are sped up from 10 to 150 times. Thus, the \emph{sparse metric} technique is a suitable choice when the network has a moderate size and a community structure with performance guarantees is desired.  
 \section{Discussion}
\vspace{-0.05in}
 \label{sec:dis}
 We have proposed two algorithms for the modularity maximization problem in complex networks. Our algorithms successfully exploit sparseness and power-degree distribution property found in many complex networks to provide performance guarantees on the solutions.  On one hand, the algorithms implied in Theorem \ref{theo:approx} are the first approximation algorithms for maximizing modularity, hence, are of theoretical interest. On the other hand, our \emph{sparse metric} approach is an efficient method to find optimal or close to optimal community structure for networks of up to thousand nodes.
 
 Fortunato and Barthelemy \cite{Fortunato07} have recently shown that in general quality functions of global defintions of community, including modularity, has an intrinsic resolution scale, known as resolution limit. Therefore, they fail to detect communities smaller than a scale, which depends on global attributes of networks such as the total size and the degree of connection among communities. However, resolution limit can be overcome by introducing a scaling parameter $\lambda >0$ into the original modularity formula as independently proposed by Arenas et al. \cite{Arenas08} and R. Lambiotte et al. \cite{Lambiotte08}.
 \[
Q_\lambda(\mathcal{C}) = \frac{1}{2m}\displaystyle\sum_{i, j} \left(A_{i,j} - \lambda \frac{d_i d_j}{2m}\right)\delta_{i,j}
\]
 
 Our proposed methods work naturally with this extension with  little modification. The only changes in the LP formulations are in the objective cofficients; the modularity matrix $B$ is replaced with a new ``multi-scale'' modularity matrix $B^\lambda$ with $B^\lambda_{i,j} = A_{i, j} - \lambda \frac{d_i d_j}{2m}$. The \emph{sparse metric} technique still applies and provides the same guarantees as solving the complete LP formulation. In addition, the constant $\lambda$ does not affect the \emph{asymptotic} approximation ratios of algorithms in Theorem \ref{theo:approx}. Our ongoing work is to design an efficient modularity approximation algorithm that both gives a better approximation ratio and perform well in practice. 

%
%

\vspace{-0.1in}

\begin{thebibliography}{10}
\providecommand{\url}[1]{#1}
\csname url@samestyle\endcsname
\providecommand{\newblock}{\relax}
\providecommand{\bibinfo}[2]{#2}
\providecommand{\BIBentrySTDinterwordspacing}{\spaceskip=0pt\relax}
\providecommand{\BIBentryALTinterwordstretchfactor}{4}
\providecommand{\BIBentryALTinterwordspacing}{\spaceskip=\fontdimen2\font plus
\BIBentryALTinterwordstretchfactor\fontdimen3\font minus
  \fontdimen4\font\relax}
\providecommand{\BIBforeignlanguage}[2]{{%
\expandafter\ifx\csname l@#1\endcsname\relax
\typeout{** WARNING: IEEEtran.bst: No hyphenation pattern has been}%
\typeout{** loaded for the language `#1'. Using the pattern for}%
\typeout{** the default language instead.}%
\else
\language=\csname l@#1\endcsname
\fi
#2}}
\providecommand{\BIBdecl}{\relax}
\BIBdecl

\bibitem{Watts98}
D.~J. Watts and S.~H. Strogatz, ``{Collective dynamics of 'small-world'
  networks},'' \emph{Nature}, vol. 393, no. 6684, 1998.

\bibitem{Barabasi00}
A.~Barabasi, R.~Albert, and H.~Jeong, ``Scale-free characteristics of random
  networks: the topology of the world-wide web,'' \emph{Physica A}, vol. 281,
  2000.

\bibitem{Milo02}
R.~Milo, S.~Shen-Orr, S.~Itzkovitz, N.~Kashtan, D.~Chklovskii, and U.~Alon,
  ``{Network motifs: simple building blocks of complex networks.}''
  \emph{Science (New York, N.Y.)}, vol. 298, no. 5594, 2002.

\bibitem{Fortunato08}
S.~Fortunato and C.~Castellano, ``Community structure in graphs,''
  \emph{Encyclopedia of Complexity and Systems Science}, 2008.

\bibitem{Girvan02}
M.~Girvan and M.~E. Newman, ``Community structure in social and biological
  networks.'' \emph{PNAS}, vol.~99, no.~12, 2002.

\bibitem{Day84}
W.~H.~E. Day and H.~Edelsbrunner, ``Efficient algorithms for agglomerative
  hierarchical clustering methods,'' \emph{Journal of Classification}, vol.~1,
  1984.

\bibitem{Reichardt06}
J.~Reichardt and S.~Bornholdt, ``Statistical mechanics of community
  detection,'' \emph{Phys. Rev. E.}, vol.~74, 2006.

\bibitem{Anca07}
A.~Gog, D.~Dumitrescu, and B.~Hirsbrunner, ``Community detection in complex
  networks using collaborative evolutionary algorithms,'' in \emph{Advances in
  Artificial Life}, ser. LNCS.\hskip 1em plus 0.5em minus 0.4em\relax Springer
  Berlin / Heidelberg, 2007, vol. 4648.

\bibitem{Duch05}
J.~Duch and A.~Arenas, ``Community detection in complex networks using extremal
  optimization,'' \emph{Phys. Rev. E}, vol.~72, no.~2, 2005.

\bibitem{Newman06}
M.~E.~J. Newman, ``{Modularity and community structure in networks},''
  \emph{Proceedings of the National Academy of Sciences}, vol. 103, no.~23,
  2006.

\bibitem{Blondel08}
V.~D. Blondel, J.-L. Guillaume, R.~Lambiotte, and E.~Lefebvre, ``{Fast
  unfolding of communities in large networks},'' \emph{Journal of Statistical
  Mechanics: Theory and Experiment}, vol. 2008, no.~10, 2008.

\bibitem{Brandes08}
U.~Brandes, D.~Delling, M.~Gaertler, R.~Gorke, M.~Hoefer, Z.~Nikoloski, and
  D.~Wagner, ``On modularity clustering,'' \emph{Knowledge and Data
  Engineering, IEEE Transactions on}, vol.~20, no.~2, 2008.

\bibitem{Agarwal08}
G.~Agarwal and D.~Kempe, ``Modularity-maximizing graph communities via
  mathematical programming,'' \emph{Eur. Phys. J. B}, vol.~66, no.~3, 2008.

\bibitem{Cafieri10}
D.~Aloise, S.~Cafieri, G.~Caporossi, P.~Hansen, S.~Perron, and L.~Liberti,
  ``Column generation algorithms for exact modularity maximization in
  networks.'' \emph{Physical Review E - Statistical, Nonlinear and Soft Matter
  Physics}, vol.~82, 2010.

\bibitem{Charikar04}
M.~Charikar and A.~Wirth, ``Maximizing quadratic programs: Extending
  grothendieck's inequality,'' \emph{FOCS}, 2004.

\bibitem{Bansal02}
N.~Bansal, A.~Blum, and S.~Chawla, ``Correlation clustering,'' in \emph{Machine
  Learning}, 2002.

\bibitem{Kemighan70}
B.~W. Kemighan and S.~Lin, ``An efficient heuristic procedure for partitioning
  graphs,'' \emph{Journal of Classification}, 1970.

\bibitem{Dantzig54}
G.~Dantzig, R.~Fulkerson, and S.~Johnson, ``Solution of a large-scale
  traveling-salesman problem,'' \emph{Operations Research}, vol.~2, 1954.

\bibitem{Applegate09}
D.~L. Applegate, R.~E. Bixby, V.~Chvátal, W.~Cook, D.~G. Espinoza,
  M.~Goycoolea, and K.~Helsgaun, ``Certification of an optimal tsp tour through
  85,900 cities,'' \emph{Operations Research Letters}, vol.~37, no.~1, 2009.

\bibitem{Nesterove97}
Y.~Nesterove, ``Semidefinite relaxation and nonconvex quadratic optimization,''
  CORE Discussion Papers 1997044, 1997.

\bibitem{Aiello:massive}
W.~Aiello, F.~Chung, and L.~Lu, ``A random graph model for massive graphs,'' in
  \emph{STOC '00}.\hskip 1em plus 0.5em minus 0.4em\relax New York, NY, USA:
  ACM, 2000.

\bibitem{Barabasi02}
A.~L. Barabási, H.~Jeong, Z.~Néda, E.~Ravasz, A.~Schubert, and T.~Vicsek,
  ``Evolution of the social network of scientific collaborations,''
  \emph{Physica A: Statistical Mechanics and its Applications}, vol. 311, 2002.

\bibitem{Albert00}
R.~Albert, H.~Jeong, and A.~Barabasi, ``Error and attack tolerance of complex
  networks,'' \emph{Nature}, vol. 406, 2000.

\bibitem{Arenas07}
D.~J. F. A. G.~S. Arenas, A, ``Size reduction of complex networks preserving
  modularity,'' \emph{New J. Phys.}, vol.~9, 2007.

\bibitem{Fortunato07}
S.~Fortunato and M.~Barthélemy, ``Resolution limit in community detection,''
  \emph{Proceedings of the National Academy of Sciences}, vol. 104, no.~1,
  2007.

\bibitem{Arenas08}
\BIBentryALTinterwordspacing
A.~Arenas, A.~Fernandez, and S.~Gomez, ``Analysis of the structure of complex
  networks at different resolution levels,'' \emph{New J. Phys.}, vol.~10,
  2008. [Online]. Available: \url{doi:10.1088/1367-2630/10/5/053039}
\BIBentrySTDinterwordspacing

\bibitem{Lambiotte08}
R.~Lambiotte, J.~C. Delvenne, and M.~Barahona, ``Laplacian dynamics and
  multiscale modular structure in networks,'' \emph{arXiv}, vol. 812, 2008.

\end{thebibliography}

\end{document}